\documentclass{llncs}
\usepackage{amsmath}
\usepackage{pstricks}
\usepackage{pst-node}
\usepackage{url}
\usepackage[all]{xy}
\usepackage{qsymbols}
\usepackage{fancybox}
\usepackage{natbib}

\newcommand{\fmdc}{feasible and more desirable choice}

\newcommand{\PX}{\ensuremath{\mathcal{P}^*(X)}}
\newcommand{\Dom}{\ensuremath{\mathcal{D}\textit{om}}}
\newcommand{\Cc}{\ensuremath{\mathcal{C}}}
\newcommand{\A}{\ensuremath{\mathcal{A}}}
\renewcommand{\S}{\ensuremath{\mathcal{S}}}
\newcommand{\feas}[1]{"-->"_{#1}}
\newcommand{\pref}[1]{".>"_{#1}}
\newcommand{\CoM}[1]{\xymatrix @C 15pt{\ar@{->}[r]&}_{#1}}
\newcommand{\node}[1]{[name=#1]{\ovalbox{#1}}}
\newcommand{\redfeas}[2]{\ncline[arrows=->,linewidth=.08,linecolor=red,linestyle=dashed]{#1}{#2}}
\newcommand{\bluefeas}[2]{\ncline[arrows=->,linewidth=.03,linecolor=blue,linestyle=dashed]{#1}{#2}}
\newcommand{\bluearcfeas}[3]{\ncarc[arrows=->,linewidth=.03,linecolor=blue,linestyle=dashed,arcangle=#3]{#1}{#2}}
\newcommand{\reddes}[2]{\ncline[arrows=->,linewidth=.08,linecolor=red,linestyle=dotted]{#1}{#2}}
\newcommand{\bluedes}[2]{\ncline[arrows=->,linewidth=.03,linecolor=blue,linestyle=dotted]{#1}{#2}}
\newcommand{\bluearcdes}[3]{\ncarc[arrows=->,linewidth=.03,linecolor=blue,linestyle=dotted,arcangle=#3]{#1}{#2}}
\newcommand{\redarcdes}[3]{\ncarc[arrows=->,linewidth=.08,linecolor=red,linestyle=dotted,arcangle=#3]{#1}{#2}}
\definecolor{vertfonce}{rgb}{0,.5,0}
\usepackage{color} \definecolor{darkbrown}{cmyk}{.3,.75,.75,.15}

\newcommand{\verte}[1]{{\color{green}#1}}
\definecolor{vertfonce}{rgb}{0,.5,0}

\title{Feasibility/Desirability Games\\
for Normal Form Games, Choice Models and Evolutionary Games}

\author{Pierre Lescanne\thanks{This research has been supported by R\'{e}gion \emph{Ile de France}.}}
\institute{Universit\'e de Lyon, ENS de Lyon, CNRS (LIP), \\
46 all\'ee d'Italie, 69364 Lyon, France}

\begin{document}

\maketitle{}

\pagestyle{plain}

\begin{abstract}
  An abstraction of normal form games is proposed, called \emph{Feasibility}/\emph{Desirability} \emph{Games} (or FD Games in short). FD Games can be seen from three points of view: as a new
  presentation of games in which Nash equilibria can be found, as choice models in microeconomics or as a model of evolution in games.
\end{abstract}

\section{Introduction}

\emph{Feasibility/Desirability Games} (FD games in short) were designed by \citet{Roux2006} as a fruitful abstraction of
\emph{normal form games}. This abstraction goes beyond the model of matrices, used for normal forms, toward this of
directed graphs, which fits better to the true nature of the problem and yields an existence theorem for a natural
abstraction of Nash equilibrium.

On another hand, behavioral economics is based on \emph{choice correspondences} which are functions owned by decision
makers and aimed to describe human behavior.  In~\cite{rubinstein98:_model_bound_ration}, Rubinstein describes the
decision process as follows:
\begin{it}
  \begin{quotation}
    \noindent an agent ... has to chose an alternative after a process of deliberation in which he answers three questions:
    \begin{itemize}
    \item ``What is feasible'?''
    \item ``What is desirable?''
    \item ``What is the best alternative according to the notion of desirability, given the feasibility constraints?''
    \end{itemize}
  \end{quotation}
\end{it}

In this article, we show that if one considers such a game as a unique decision maker and if the relation of the game
called \emph{\fmdc{}} is acyclic (i.e., with no path from a node to itself), then the function that yields the abstract
Nash equilibria is a choice correspondence.   We show how an actual decision maker can be made, namely it can be
implemented as a game and the choices he makes are abstract Nash equilibria.  Hence we say that a decision maker can be
constituted of many agents, which is not a surprise.  Think of a decision taken by a directorate.  The decision is a the
rest of a compromise, i.e., an equilibrium.

\citet{SalantRubi08} give a sample of choice functions, based on what the authors call a \emph{frame}.  We show that FD
games are frames, in other words, a game with agents is a good tool to simulate a decision maker.  This is also related
to what \citet{bernheim07:_towar_choic_theret_found_for} call \emph{choice with ancillary functions} and
\cite{manzini07:_sequen_ration_choic}, \emph{Rational Shortlist Method}.  FD games have also interesting connections
with evolutionary games.  Let us mention that \emph{choice functions} have been studied by \citet{Nehring97}, but
despite it speaks about relations, this approach is not directly relevant with the one presented in this paper.  See
also \citet{alcantud107} for a somewhat more restricted extension of normal form games.

In this article, we present first in Section~\ref{sec:fd-games} the concept of \emph{FD~game} followed by this of
equilibrium: \emph{abstract Nash equilibrium} (Section~\ref{sec:ane}) and its extension called \emph{FD-equilibrium}
(Section~\ref{sec:FDeq}). We recall in Section~\ref{sec:choice} the notion of \emph{choice correspondence} and show its
connection with FD-games. Two kinds of choice correspondences are described in Section~\ref{sec:EqAN_as_ch}.
\emph{Evolutionary games} are described as FD-games in Section~\ref{sec:evo}.



\section{FD games}
\label{sec:fd-games}


FD games\footnote{Under the name of CP games.} were conceived as extensions of strategic games
(\cite{roux08:_conver_prefer_games}) which intend to algebraically formalize games using the minimal set of concepts
(Occam's Razor).  They actually implement the three items of the ``deliberation process'' as presented in
\cite{rubinstein98:_model_bound_ration,rubinstein06:microec}, namely \emph{feasibility}, \emph{desirability}, and
\emph{choice of the most desirable among the feasible alternatives}. Like in strategic games, there is a set $\A$ of
\emph{agents} and a set $\S$ of \emph{situations}\footnote{This corresponds to what is called a strategy profile in
  strategic games.  Sometimes a ``situation'' is called an ``outcome''.}.  Each agent moves from a situation to another;
but the move can occur only if it is possible for the agent, this is what
\cite{rubinstein98:_model_bound_ration,rubinstein06:microec} calls \emph{feasibility}. We write it~$\feas{a}$. To be
precise, $s \feas{a} s'$ means that the move from $s$ to $s'$ is feasible for agent $a$.  An agent can have a desire to
go to a situation rather than staying where he is, this is what is called \emph{desirability}
in~\cite{rubinstein98:_model_bound_ration,rubinstein06:microec}; it is written~$\pref{a}$.  $s\pref{a}s'$ means that
agent $a$ who is in situation $s$ desires to go in situation $s'$.  For instance, if the FD-game is a game in normal
form and $s$ and $s'$ are strategy profiles, strategy profile $s'$ has a better payoff for him than strategy profile
$s'$.  Hence, formally, an FD game is a \emph{4-uple} $`G = \langle \A, \S, (\feas{a})_{a\in \A}, (\pref{a})_{a\in
  \A}\rangle$.

\begin{example}
  Consider the game with two agents, thin blue and fat red and four situations $(H,H), (H,T), (T,T), (T,H)$.  Thin blue agent has feasibility
  \raisebox{2pt}{$\begin{psmatrix}[colsep=20pt][name=a]&[name=b]\ncline[arrows=->,linewidth=.03,linecolor=blue,linestyle=dashed]{a}{b}\end{psmatrix}$} and desirability
  \raisebox{2pt}{$\begin{psmatrix}[colsep=20pt][name=a]&[name=b]\ncline[arrows=->,linewidth=.03,linecolor=blue,linestyle=dotted]{a}{b}\end{psmatrix}$}, fat red agent has feasibility
  \raisebox{2pt}{$\begin{psmatrix}[colsep=25pt][name=a]&[name=b]\ncline[arrows=->,linewidth=.08,linecolor=red,linestyle=dashed]{a}{b}\end{psmatrix}$} and desirability
  \raisebox{2pt}{$
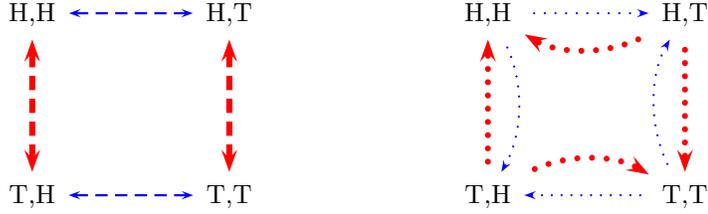
\begin{psmatrix}[colsep=25pt][name=a]&[name=b]\ncline[arrows=->,linewidth=.08,linecolor=red,linestyle=dotted]{a}{b}\end{psmatrix}$},
\end{example}

\medskip

\hspace*{-20pt}
  \begin{psmatrix}[colsep=2cm,rowsep=2cm]
    & [name=HH]{H,H} & [name=HT]{H,T}\\
    & [name=TH]{T,H} & [name=TT]{T,T} 
    \psset{nodesep=5pt} %
    \ncline[arrows=<->,linewidth=.03,linecolor=blue,linestyle=dashed]{HH}{HT}
    \ncline[arrows=<->,linewidth=.03,linecolor=blue,linestyle=dashed]{TH}{TT} %
    \ncline[arrows=<->,linewidth=.08,linecolor=red,linestyle=dashed]{HH}{TH}
    \ncline[arrows=<->,linewidth=.08,linecolor=red,linestyle=dashed]{HT}{TT}
  \end{psmatrix}
  \qquad
  \begin{psmatrix}[colsep=2cm,rowsep=2cm]
    & [name=HH]{H,H} & [name=HT]{H,T}\\
    & [name=TH]{T,H} & [name=TT]{T,T} \psset{nodesep=5pt} %
    \ncline[arrows=->,linewidth=.03,linecolor=blue,linestyle=dotted]{HH}{HT} %
    \ncline[arrows=->,linewidth=.03,linecolor=blue,linestyle=dotted]{TT}{TH} %
    \ncline[arrows=->,linewidth=.08,linecolor=red,linestyle=dotted]{TH}{HH} %
    \ncline[arrows=->,linewidth=.08,linecolor=red,linestyle=dotted]{HT}{TT} %
    \ncarc[arrows=->,linewidth=.08,linecolor=red,linestyle=dotted,arcangle=30]{HT}{HH} %
    \ncarc[arrows=->,linewidth=.03,linecolor=blue,linestyle=dotted,arcangle=30]{HH}{TH} 
    \ncarc[arrows=->,linewidth=.08,linecolor=red,linestyle=dotted,arcangle=30]{TH}{TT}
    \ncarc[arrows=->,linewidth=.03,linecolor=blue,linestyle=dotted,arcangle=30]{TT}{HT}
  \end{psmatrix}

\medskip

When presented as a strategic game, this game is known as \emph{matching pennies}.  The coming example is a game which is not derived from a strategic game.

\begin{example}[A quest for the wonderland]\label{ex:quest} Two captains with their ships, which are not sister ships,
 and their crew look for nice islands to stay.  Due to the wind and the different performances of the ships, the eight islands are not equally accessible and
  some are accessible one way by a ship, but not the other way.  Similarly not all the islands are equally nice and the crew may prefer (desire) one over
  the other.  This is represented by the following diagrams.
\[
\begin{psmatrix}[colsep=25pt,rowsep=15pt] &&&\node{D} & \node{H}\\ &&\node{B} & \node{E}\\ &\node{A} && \node{F} \\ &&\node{C} \\ &&&\node{G} \redfeas{A}{B} %
\redfeas{B}{D} %
\redfeas{C}{B} \bluefeas{C}{F} \redfeas{F}{G} \bluefeas{A}{C} \bluefeas{G}{C} \redfeas{H}{E} \bluefeas{B}{E} \bluearcfeas{B}{C}{30} \bluearcfeas{E}{H}{30} \bluearcfeas{H}{E}{30}
\end{psmatrix} \qquad
\begin{psmatrix}[colsep=25pt,rowsep=15pt] &&&\node{D} & \node{H}\\ &&\node{B} & \node{E}\\ &\node{A} && \node{F} \\ &&\node{C} \\ &&&\node{G}
\end{psmatrix} %
\reddes{A}{B} %
\reddes{C}{B} %
\reddes{B}{D} %
\bluedes{A}{C}%
\redarcdes{B}{E}{-30} %
\bluearcdes{B}{E}{30} %
\bluearcdes{E}{H}{30} %
\redarcdes{H}{E}{30}
\redarcdes{C}{G}{30} %
\bluearcdes{G}{C}{30} \reddes{F}{G} \bluedes{C}{F}
\]
\begin{center}
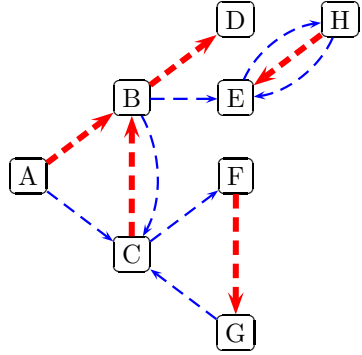
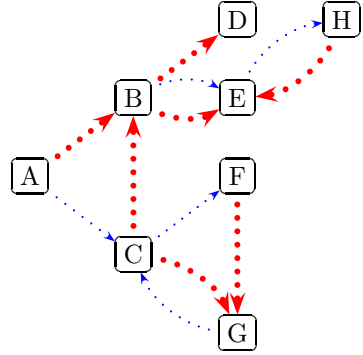
  \textit{Feasibility of the journeys}\qquad \qquad \qquad \textit{Desirability among islands}
\end{center}
\end{example}

\section{Equilibria}
\label{sec:equilib}

Let us now be a little formal.

\subsection{Abstract Nash Equilibrium }
\label{sec:ane}

A Nash equilibrium is a situation in which no agent can move toward a situation he (she) desires.

\begin{definition}[Abstract Nash equilibrium] An \emph{abstract Nash equilibrium} is a situation $s$ such that: \( `A a`:\A, `A s'`:\S \ . \ s \feas{a} s' \wedge s \pref{a} s' \
\Longrightarrow\ s=s'.\) This is written $\textrm{Eq}^{aN}_{`G}(s)$ (aN stands for \emph{abstract Nash}).
\end{definition}

This suggests a relation $\rightarrow_{a}$ called \emph{\fmdc{} for $a$} defined as
\[\rightarrow_a \quad = \quad \feas{a} \cap \pref{a}\] i.e. $s"->"_a s'$ if $s\feas{a}s'$ and $s\pref{a}s'$.  From these relations, one defines a relation \emph{\fmdc{}} which sums
up the \emph{\fmdc} of all the agents).
\[\rightarrow \quad = \quad \bigcup_{a\in\A} "->"_a.\]
\begin{definition}[Abstract Nash equilibrium, as sink] A situation $s$ is an \emph{abstract Nash equilibria} if $s$ is a sink\footnote{Given a relation $R$ (or a graph) a
    \emph{sink} is a node $s$ such that $`A s', s"->"s' "=>" s=s'$.} for the \fmdc{}, i.e., if \(`A s' \in \S, s "->" s' \ "=>" s = s'\).
\end{definition}
\emph{Abstract Nash equilibria} can be seen as one of  the \emph{feasible most desirable choices}. Actually there is no other situation which is feasible and more desirable.

This is illustrated by the \emph{prisoner's dilemma}.

\begin{example}[Prisoner's dilemma]
The situations are \emph{Quits} ($Q$) and \emph{Finks} ($F$).

\vspace*{20pt}

  \hspace*{-50pt} \(
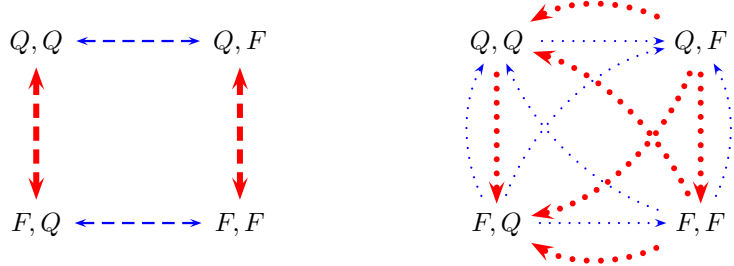
  \begin{psmatrix}[colsep=2cm,rowsep=2cm] &[name=QQ]{Q,Q} & [name=QF]{Q,F}\\ &[name=FQ]{F,Q} & [name=FF]{F,F} \psset{nodesep=5pt}
\ncline[arrows=<->,linewidth=.03,linecolor=blue,linestyle=dashed]{QQ}{QF} \ncline[arrows=<->,linewidth=.03,linecolor=blue,linestyle=dashed]{FQ}{FF}
\ncline[arrows=<->,linewidth=.08,linecolor=red,linestyle=dashed]{QQ}{FQ} \ncline[arrows=<->,linewidth=.08,linecolor=red,linestyle=dashed]{QF}{FF}
  \end{psmatrix} \qquad
  \begin{psmatrix}[colsep=2cm,rowsep=2cm] &[name=QQ]{Q,Q} & [name=QF]{Q,F}\\ &[name=FQ]{F,Q} & [name=FF]{F,F} \psset{nodesep=3.5pt}
\ncline[arrows=->,linewidth=.03,linecolor=blue,linestyle=dotted]{QQ}{QF} \ncarc[arrows=->,linewidth=.08,linecolor=red,linestyle=dotted, arcangle=-30]{QF}{QQ}
\ncline[arrows=->,linewidth=.03,linecolor=blue,linestyle=dotted]{FQ}{FF} \ncarc[arrows=->,linewidth=.08,linecolor=red,linestyle=dotted, arcangle=30]{FF}{FQ}
\ncline[arrows=->,linewidth=.08,linecolor=red,linestyle=dotted]{QQ}{FQ} \ncarc[arrows=->,linewidth=.03,linecolor=blue,linestyle=dotted, arcangle=30]{FQ}{QQ}
\ncline[arrows=->,linewidth=.08,linecolor=red,linestyle=dotted]{QF}{FF} \ncarc[arrows=->,linewidth=.03,linecolor=blue,linestyle=dotted, arcangle=-30]{FF}{QF}
\ncarc[arrows=->,linewidth=.08,linecolor=red,linestyle=dotted, arcangle=-25]{FF}{QQ} \ncarc[arrows=->,linewidth=.03,linecolor=blue,linestyle=dotted, arcangle=25]{FF}{QQ}
\ncarc[arrows=->,linewidth=.08,linecolor=red,linestyle=dotted, arcangle=30]{QF}{FQ} \ncarc[arrows=->,linewidth=.03,linecolor=blue,linestyle=dotted, arcangle=30]{FQ}{QF}
  \end{psmatrix} \qquad \)

  \bigskip\medskip

  \noindent The above graphs represent the \emph{feasibility} (on the left) and the \emph{desirability} (on the right).  Vertical arrows for feasibility represent possibilities for
  the first agent to change his (her) attitude and horizontal arrows are for the second agent.  The right digram represents abstractly (i.e., without referring to payoffs that are
  meaningless in that case) the desires (preferences) of the agents.  Both agents prefer $(Q,Q)$ to $(F,F)$ (desire $(Q,Q)$ more than $(F,F)$), but first agent prefers $(F,Q)$ to
  $(F,F)$, $(F,Q)$ to $(Q,F)$, etc.  The graphs below represent the \emph{\fmdc} for both prisoners and the \emph{(general) \fmdc}:

  \bigskip

  \hspace*{-50pt} \(
  \begin{psmatrix}[colsep=2cm,rowsep=2cm] &[name=QQ]{Q,Q} & [name=QF]{Q,F}\\ &[name=FQ]{F,Q} & [name=FF]{F,F} \psset{nodesep=5pt}
\ncline[arrows=->,linewidth=.03,linecolor=blue]{QQ}{QF} \ncline[arrows=->,linewidth=.03,linecolor=blue]{FQ}{FF} \ncline[arrows=->,linewidth=.08,linecolor=red]{QQ}{FQ}
\ncline[arrows=->,linewidth=.08,linecolor=red]{QF}{FF}
  \end{psmatrix} \quad
  \begin{psmatrix}[colsep=2cm,rowsep=2cm] &[name=QQ]{Q,Q} & [name=QF]{Q,F}\\ &[name=FQ]{F,Q} & [name=FF]{\mathbf{\mathit{\color{green}F,F}}} \psset{nodesep=5pt}
\ncline[arrows=->,linewidth=.03]{QQ}{QF} \ncline[arrows=->,linewidth=.03]{FQ}{FF} \ncline[arrows=->,linewidth=.03]{QQ}{FQ} \ncline[arrows=->,linewidth=.03]{QF}{FF}
  \end{psmatrix} \)

  \bigskip \newcommand{\flch}{
    \begin{psmatrix}[colsep=.4cm,rowsep=0cm] [name=a] & [name=b] \psset{nodesep=.1pt} \ncline[arrows=->,linewidth=.02]{a}{b}
    \end{psmatrix} }

  \noindent Clearly $F,F$ is a sink for \raisebox{3pt}{$\flch$} and this is an abstract Nash equilibrium; this means that there is no more desirable choice for the prisoners in
  this situation.  One notices that the change from $(F,F)$ to $(Q,Q)$ which corresponds to a more desirable choice for both agents is not feasible due to constraints which forces
  feasibility to be on horizontal or vertical lines only.
\end{example} 
The \emph{Battle of the Sexes}, also called \emph{Bach or Stravinsky} or \emph{BoS} (\cite{osborne94:_cours_game_theory,osborne04a}) is an example with two abstract
Nash equilibria.

\begin{example}[Battles of the Sexes] It is given by the two diagrams which describe \emph{feasibility} and \emph{desirability}:

  \bigskip

  \hspace*{-50pt} \(
  \begin{psmatrix}[colsep=2cm,rowsep=2cm] &[name=BB]{B,B} & [name=BS]{B,S}\\ &[name=SB]{S,B} & [name=SS]{S,S} \psset{nodesep=5pt}
    \ncline[arrows=<->,linewidth=.03,linecolor=blue,linestyle=dashed]{BB}{BS} \ncline[arrows=<->,linewidth=.03,linecolor=blue,linestyle=dashed]{SB}{SS}
    \ncline[arrows=<->,linewidth=.08,linecolor=red,linestyle=dashed]{BB}{SB} \ncline[arrows=<->,linewidth=.08,linecolor=red,linestyle=dashed]{BS}{SS}
  \end{psmatrix} \qquad
  \begin{psmatrix}[colsep=2cm,rowsep=2cm] &[name=BB]{B,B} & [name=BS]{B,S}\\ &[name=SB]{S,B} & [name=SS]{S,S} \psset{nodesep=3.5pt}
    \ncline[arrows=<-,linewidth=.03,linecolor=blue,linestyle=dotted]{BB}{BS} \ncarc[arrows=->,linewidth=.08,linecolor=red,linestyle=dotted, arcangle=-30]{BS}{BB}
    \ncline[arrows=->,linewidth=.03,linecolor=blue,linestyle=dotted]{SB}{SS} %
    \ncarc[arrows=<-,linewidth=.08,linecolor=red,linestyle=dotted, arcangle=30]{SS}{SB} %
    \ncline[arrows=<-,linewidth=.08,linecolor=red,linestyle=dotted]{BB}{SB} \ncarc[arrows=->,linewidth=.03,linecolor=blue,linestyle=dotted, arcangle=30]{SB}{BB} %
    \ncline[arrows=->,linewidth=.08,linecolor=red,linestyle=dotted]{BS}{SS} \ncarc[arrows=->,linewidth=.03,linecolor=blue,linestyle=dotted, arcangle=30]{BS}{SS}
    \ncarc[arrows=->,linewidth=.03,linecolor=blue,linestyle=dotted, arcangle=-30]{SS}{BB} \ncarc[arrows=->,linewidth=.08,linecolor=red,linestyle=dotted, arcangle=-30]{BB}{SS}
  \end{psmatrix} \qquad \)

  \bigskip\medskip

  \noindent with the \emph{\fmdc} for both agents and the \emph{(general) \fmdc{}}:

  \bigskip

  \hspace*{-50pt} \(
  \begin{psmatrix}[colsep=2cm,rowsep=2cm] &[name=BB]{B,B} & [name=BS]{B,S}\\ &[name=SB]{S,B} & [name=SS]{S,S} \psset{nodesep=5pt}
    \ncline[arrows=<-,linewidth=.03,linecolor=blue]{BB}{BS} %
    \ncline[arrows=->,linewidth=.03,linecolor=blue]{SB}{SS} %
    \ncline[arrows=<-,linewidth=.08,linecolor=red]{BB}{SB}
    \ncline[arrows=->,linewidth=.08,linecolor=red]{BS}{SS}
  \end{psmatrix}
  \quad
  \begin{psmatrix}[colsep=2cm,rowsep=2cm] &[name=BB]{\color{green}B,B} & [name=BS]{B,S}\\ &[name=SB]{S,B} & [name=SS]{\mathbf{\mathit{\color{green}S,S}}}
    \psset{nodesep=5pt} \ncline[arrows=<-,linewidth=.03]{BB}{BS} %
    \ncline[arrows=->,linewidth=.03]{SB}{SS} %
    \ncline[arrows=<-,linewidth=.03]{BB}{SB} %
    \ncline[arrows=->,linewidth=.03]{BS}{SS}
  \end{psmatrix} \)

  \bigskip
There are two Nash equilibria ${\color{green}B,B}$ and ${\color{green}S,S}$.
\end{example}

\begin{example}[A quest for the wonderland, the choice] In Example~\ref{ex:quest} the \emph{\fmdc} is given by the following diagram:
  
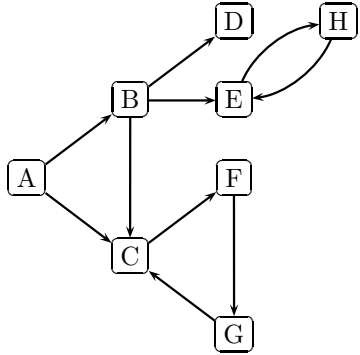
\[
\begin{psmatrix}[colsep=25pt,rowsep=15pt] 
&&&\node{D} & \node{H}\\ 
&&\node{B} & \node{E}\\ 
&\node{A} && \node{F} \\ 
&&\node{C} \\ 
&&&\node{G} %
\ncline[arrows=->,linewidth=.03]{A}{B}
\ncline[arrows=->,linewidth=.03]{A}{C} %
\ncline[arrows=->,linewidth=.03]{B}{C} %
\ncline[arrows=->,linewidth=.03]{C}{F} %
\ncline[arrows=->,linewidth=.03]{F}{G} %
\ncline[arrows=->,linewidth=.03]{G}{C} %
\ncline[arrows=->,linewidth=.03]{B}{D} %
\ncline[arrows=->,linewidth=.03]{B}{E} %
\ncarc[arrows=->,linewidth=.03,arcangle=30]{E}{H} %
\ncarc[arrows=->,linewidth=.03,arcangle=30]{H}{E} %
\end{psmatrix}
\]
It has an abstract Nash equilibrium, namely $D$, i.e., one the feasible most desirable choice.
\end{example} 

Actually the relation \emph{\fmdc} is the relation of interest, but the two relations \emph{feasibility} and \emph{desirability} are essential for two reasons. First they make the
connection with strategic games (games in normal forms), in particular \emph{desirability} is what is called \emph{preference} in strategic
games (\cite{osborne94:_cours_game_theory}) whereas \emph{feasibility} is an abstraction of the relation that allows moves only along rows, columns, heights, etc. in strategic games.
Second they are methodologically important, as they allow thinking models in terms of two relations answering two specific questions: what is feasible? what is desirable?

In algorithmic game theory (\cite{daskalakis09:_compl_of_comput_nash_equil,nisan07:_algor_game_theor,johnson07:_np_compl_colum}), Nash equilibria are also sinks for a relation, but
the definition of a Nash equilibrium is not given as being a sink.  A mixed strategy Nash equilibrium is actually a sink for a completely different relation which is used to find
it and which has no connection with the \emph{\fmdc}.

\subsection{FD equilibria}
\label{sec:FDeq}

Let us recall few concepts on directed graphs and strongly connected components.

\subsubsection{Strongly connected components.}
\label{sec:scc}

Given a directed graph, a strongly connected component (a SCC in short) is a maximal set of nodes connected in both direction by paths\footnote{A \emph{path} is a sequence of
  arcs.}. In other words a SCC is a set of nodes of the directed graph, which are connected by paths (sequences of arcs).  Moreover this set is maximal, which means that if one
adds a node, then there is a node in the SCC which is not connected to it or the other way around.  Fig.~\ref{fig:graph_scc} shows a graph, its decomposition into strongly
connected components and its reduced graph. Actually it has four strongly connected components associated with ${\color{green} \bullet}$, ${\color{brown} \blacksquare}$,
${\color{red} \bigstar}$ and ${\color{blue} \blacktriangle}$.  The SCC's form a graph called the \emph{reduced graph}, the nodes of which are the SCC's and the arc of which are as
follows: there is an arc from a SCC $N$ to a SCC $N'$ if there are a node $n$ in $N$, a node $n'$ in $N'$ and an arc $n"->"n'$.

\begin{figure}[ht]
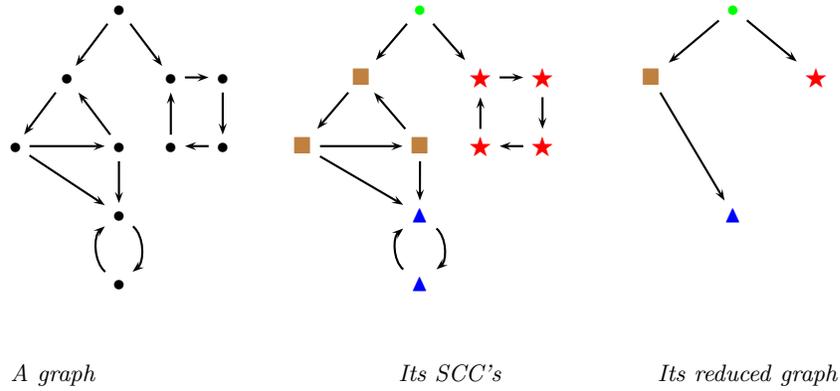
 \centering \(
  \begin{psmatrix}
    [nodesep=3pt,colsep=15pt,rowsep=15pt] &&& [name=a]{\bullet} \\ %
    && [name=c]{\bullet} && [name=d]{\bullet}& [name=j]{\bullet}\\ %
    &[name=e]{\bullet} && [name=f]{\bullet} & [name=k]{\bullet} & [name=l]{\bullet}\\ %
    &&& [name=h]{\bullet}\\ &&& [name=i]{\bullet} \\ %
    \ncline[arrows=->]{a}{c} \ncline[arrows=->]{a}{d} %
    \ncline[arrows=->]{c}{e} \ncline[arrows=->]{e}{f} %
    \ncline[arrows=->]{f}{c} \ncline[arrows=->]{e}{g} %
    \ncline[arrows=->]{f}{g} \ncline[arrows=->]{e}{h} %
    \ncline[arrows=->]{f}{h}%
    \ncarc[arrows=->,arcangle=50]{i}{h} %
    \ncarc[arrows=->,arcangle=50]{h}{i} %
    \ncline[arrows=->]{d}{j} \ncline[arrows=->]{j}{l} %
    \ncline[arrows=->]{l}{k} \ncline[arrows=->]{k}{d}
  \end{psmatrix} 
  \quad
  \begin{psmatrix}[nodesep=3pt,colsep=15pt,rowsep=15pt] &&& [name=a]{\color{green} \bullet} \\ && [name=c]{\color{brown} \blacksquare} && [name=d]{\color{red} \bigstar}&
[name=j]{\color{red} \bigstar}\\ &[name=e]{\color{brown} \blacksquare} && [name=f]{\color{brown} \blacksquare} & [name=k]{\color{red} \bigstar} & [name=l]{\color{red} \bigstar}\\
&&& [name=h]{\color{blue} \blacktriangle}\\ &&& [name=i]{\color{blue} \blacktriangle} \\ \ncline[arrows=->]{a}{c} \ncline[arrows=->]{a}{d} \ncline[arrows=->]{c}{e}
\ncline[arrows=->]{e}{f} \ncline[arrows=->]{f}{c} \ncline[arrows=->]{e}{g} \ncline[arrows=->]{f}{g} \ncline[arrows=->]{e}{h} \ncline[arrows=->]{f}{h}
\ncarc[arrows=->,arcangle=50]{i}{h} \ncarc[arrows=->,arcangle=50]{h}{i} \ncline[arrows=->]{d}{j} \ncline[arrows=->]{j}{l} \ncline[arrows=->]{l}{k} \ncline[arrows=->]{k}{d}
  \end{psmatrix} \quad
  \begin{psmatrix}[nodesep=3pt,colsep=24pt,rowsep=15pt] && [name=a]{\color{green} \bullet} \\ & [name=b]{\color{brown} \blacksquare} && [name=c]{\color{red} \bigstar} \\\\ &&
[name=d]{\color{blue} \blacktriangle} \\\\ \ncline[arrows=->]{a}{b} \ncline[arrows=->]{a}{c} \ncline[arrows=->]{b}{d}
  \end{psmatrix} \) \hspace*{20pt}\emph{A graph} \qquad \qquad \qquad\qquad \qquad\qquad \emph{Its SCC's} \qquad \qquad \qquad \emph{Its reduced graph}
  \caption{A graph, its SCC's and, its reduced graph}
  \label{fig:graph_scc}
\end{figure}

\subsubsection{FD equilibria as SCC's.}
\label{sec:FDe_as_SCC}

An abstract Nash equilibrium is a sink in the graph of the \emph{\fmdc}.  We know that abstract  Nash equilibria do not exist always, but if we generalize the concept, there exists always
an equilibrium.  To generalize abstract Nash equilibria, we consider the case where the agents have reached, by the \emph{\fmdc}, some kind of end point (end ``cluster''), actually an SCC, and are
unable to leave it.  This corresponds to what people call a \emph{dynamic equilibrium}.  Hence a natural extension is to say that a \emph{FD~equilibrium} is a sink in the reduced
graph, i.e. a SCC with no out arc. The two SCC's associated with ${\color{blue} \blacktriangle}$ and ${\color{red} \bigstar}$ in Fig.~\ref{fig:graph_scc} are such sinks in the
reduced graph.

\subsubsection{Examples of FD equilibria}
\label{sec:ex}

\begin{example}[Matching pennies]
  The \emph{\fmdc{}} for the matching pennies is given by the following diagrams.

\medskip

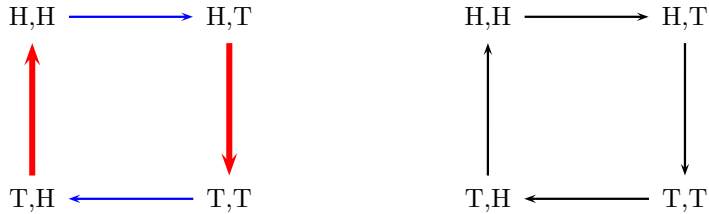
\hspace*{-20pt}
  \begin{psmatrix}[colsep=2cm,rowsep=2cm]
    & [name=HH]{H,H} & [name=HT]{H,T}\\
    & [name=TH]{T,H} & [name=TT]{T,T} 
    \psset{nodesep=5pt} %
    \ncline[arrows=->,linewidth=.03,linecolor=blue]{HH}{HT}
    \ncline[arrows=<-,linewidth=.03,linecolor=blue]{TH}{TT} %
    \ncline[arrows=<-,linewidth=.08,linecolor=red]{HH}{TH}
    \ncline[arrows=->,linewidth=.08,linecolor=red]{HT}{TT}
  \end{psmatrix}
  \qquad
  \begin{psmatrix}[colsep=2cm,rowsep=2cm]
    & [name=HH]{H,H} & [name=HT]{H,T}\\
    & [name=TH]{T,H} & [name=TT]{T,T} 
    \psset{nodesep=5pt} %
    \ncline[arrows=->,linewidth=.03]{HH}{HT}
    \ncline[arrows=<-,linewidth=.03]{TH}{TT} %
    \ncline[arrows=<-,linewidth=.03]{HH}{TH}
    \ncline[arrows=->,linewidth=.03]{HT}{TT}
  \end{psmatrix}\end{example}
\begin{example}[Matching pennies with hidden coins] 
  We imagine a new version of \emph{matching pennies} where players have three actions \emph{Presenting the Coin Head up} ($H$), \emph{Presenting the Coin Tail up} ($T$) or
  \emph{Not Presenting the Coin} ($N$). When both present the same side, fat red wins and thin blue loses and when they both present different sides, thin blue wins and fat red
  loses. In addition, we assume that when both present nothing they both win and when only one presents nothing they both lose.  The game is presented in
  Fig.~\ref{fig:match_pennies_2}.  The first diagram is \emph{feasibility}, the second diagram is \emph{desirability}\footnote{There are two options.  You can consider that the
    desirability is not transitive or you can assume that the arrows that can be deduced by transitivity have been dropped. But this difference is unimportant in what follows.},
  the third diagram is \emph{general \fmdc}, and the fourth diagram is the reduced graph of \emph{general \fmdc}.  One notices two FD equilibria. First, $(N,N)$ which is also an
  abstract Nash equilibrium and $\{(H,H), (H,T), (T,H), (T,T)\}$ which corresponds to a SCC made of four situations.
\end{example}
\begin{figure}[ht]
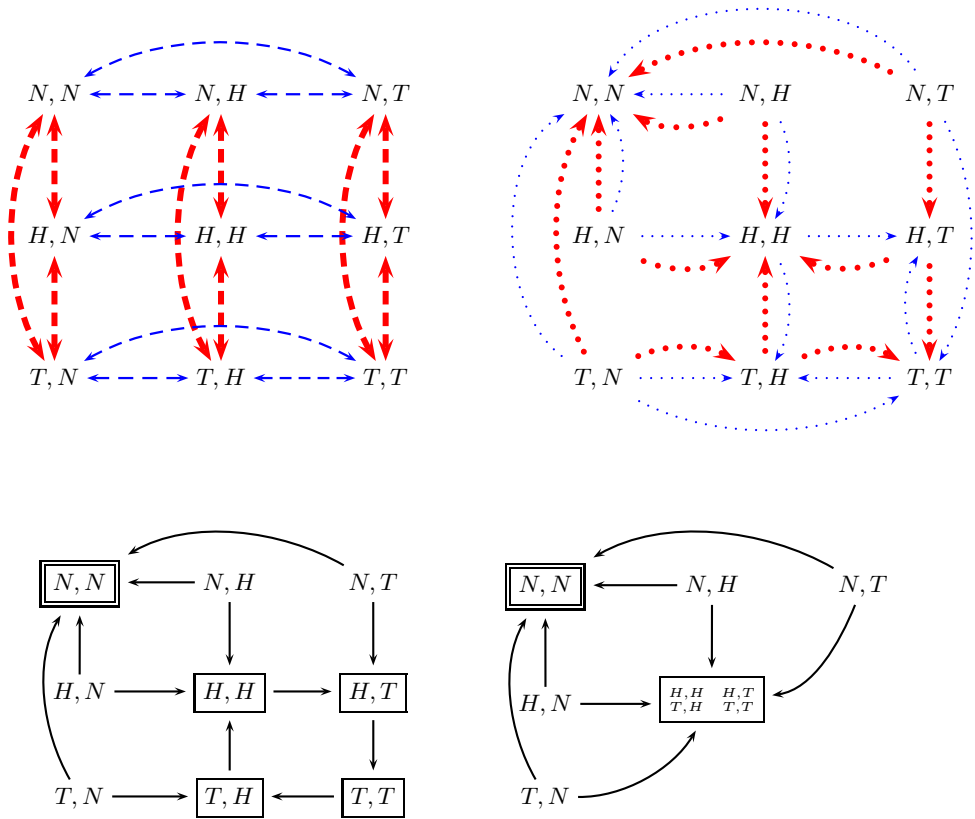
 \centering
\[ \hspace*{-40pt}
\begin{psmatrix}[nodesep=3pt] & [name=NN]{N,N} & [name=NH]{N,H} & [name=NT]{N,T}\\ & [name=HN]{H,N} & [name=HH]{H,H} & [name=HT]{H,T}\\ & [name=TN]{T,N} & [name=TH]{T,H} &
[name=TT]{T,T} \ncline[arrows=<->,linewidth=.08,linecolor=red,linestyle=dashed]{NN}{HN} \ncline[arrows=<->,linewidth=.08,linecolor=red,linestyle=dashed]{HN}{TN}
\ncarc[arrows=<->,linewidth=.08,linecolor=red,linestyle=dashed,arcangle=-30]{NN}{TN} \ncline[arrows=<->,linewidth=.08,linecolor=red,linestyle=dashed]{NH}{HH}
\ncline[arrows=<->,linewidth=.08,linecolor=red,linestyle=dashed]{HH}{TH} \ncarc[arrows=<->,linewidth=.08,linecolor=red,linestyle=dashed,arcangle=-30]{NH}{TH}
\ncline[arrows=<->,linewidth=.08,linecolor=red,linestyle=dashed]{NT}{HT} \ncline[arrows=<->,linewidth=.08,linecolor=red,linestyle=dashed]{HT}{TT}
\ncarc[arrows=<->,linewidth=.08,linecolor=red,linestyle=dashed,arcangle=-30]{NT}{TT} 
\ncline[arrows=<->,linewidth=.03,linecolor=blue,linestyle=dashed]{NN}{NH}
\ncline[arrows=<->,linewidth=.03,linecolor=blue,linestyle=dashed]{NH}{NT} \ncarc[arrows=<->,linewidth=.03,linecolor=blue,linestyle=dashed,arcangle=30]{NN}{NT}
\ncline[arrows=<->,linewidth=.03,linecolor=blue,linestyle=dashed]{HN}{HH} \ncline[arrows=<->,linewidth=.03,linecolor=blue,linestyle=dashed]{HH}{HT}
\ncarc[arrows=<->,linewidth=.03,linecolor=blue,linestyle=dashed,arcangle=30]{HN}{HT} \ncline[arrows=<->,linewidth=.03,linecolor=blue,linestyle=dashed]{TN}{TH} %
\ncline[arrows=<->,linewidth=.03,linecolor=blue,linestyle=dashed]{TH}{TT} \ncarc[arrows=<->,linewidth=.03,linecolor=blue,linestyle=dashed,arcangle=30]{TN}{TT}
   \end{psmatrix} 
\qquad
\begin{psmatrix}[nodesep=3pt] 
& [name=NN]{N,N} & [name=NH]{N,H} & [name=NT]{N,T}\\ 
& [name=HN]{H,N} & [name=HH]{H,H} & [name=HT]{H,T}\\ 
& [name=TN]{T,N} & [name=TH]{T,H} & [name=TT]{T,T} 
\ncline[arrows=->,linewidth=.08,linecolor=red,linestyle=dotted]{HN}{NN} \ncline[arrows=->,linewidth=.08,linecolor=red,linestyle=dotted]{NH}{HH}
\ncline[arrows=->,linewidth=.08,linecolor=red,linestyle=dotted]{TH}{HH} \ncline[arrows=->,linewidth=.08,linecolor=red,linestyle=dotted]{NT}{HT}
\ncline[arrows=->,linewidth=.08,linecolor=red,linestyle=dotted]{HT}{TT} \ncarc[arrows=->,linewidth=.08,linecolor=red,linestyle=dotted,arcangle=30]{TN}{NN}
\ncarc[arrows=->,linewidth=.08,linecolor=red,linestyle=dotted,arcangle=-30]{NT}{NN} \ncarc[arrows=->,linewidth=.08,linecolor=red,linestyle=dotted,arcangle=30]{HT}{HH}
\ncarc[arrows=->,linewidth=.08,linecolor=red,linestyle=dotted,arcangle=30]{TH}{TT} \ncarc[arrows=->,linewidth=.08,linecolor=red,linestyle=dotted,arcangle=30]{NH}{NN}
\ncarc[arrows=->,linewidth=.08,linecolor=red,linestyle=dotted,arcangle=-30]{HN}{HH} \ncarc[arrows=->,linewidth=.08,linecolor=red,linestyle=dotted,arcangle=30]{TN}{TH} 
\ncline[arrows=->,linewidth=.03,linecolor=blue,linestyle=dotted]{NH}{NN} \ncline[arrows=->,linewidth=.03,linecolor=blue,linestyle=dotted]{HN}{HH}
\ncline[arrows=->,linewidth=.03,linecolor=blue,linestyle=dotted]{HH}{HT} \ncarc[arrows=->,linewidth=.03,linecolor=blue,linestyle=dotted,arcangle=30]{HH}{TH}
\ncline[arrows=->,linewidth=.03,linecolor=blue,linestyle=dotted]{TT}{TH} \ncarc[arrows=->,linewidth=.03,linecolor=blue,linestyle=dotted,arcangle=30]{TT}{HT}
\ncarc[arrows=->,linewidth=.03,linecolor=blue,linestyle=dotted,arcangle=60]{TN}{NN} \ncarc[arrows=->,linewidth=.03,linecolor=blue,linestyle=dotted,arcangle=-60]{NT}{NN}
\ncarc[arrows=->,linewidth=.03,linecolor=blue,linestyle=dotted,arcangle=-30]{HN}{NN} \ncarc[arrows=->,linewidth=.03,linecolor=blue,linestyle=dotted,arcangle=30]{NH}{HH}
\ncarc[arrows=->,linewidth=.03,linecolor=blue,linestyle=dotted,arcangle=-30]{TN}{TT} \ncarc[arrows=->,linewidth=.03,linecolor=blue,linestyle=dotted,arcangle=30]{NT}{TT}
\ncline[arrows=->,linewidth=.03,linecolor=blue,linestyle=dotted]{TN}{TH}
 \end{psmatrix}
\]

\vspace*{40pt}

\[ \hspace*{-40pt}
\begin{psmatrix}[nodesep=3pt,colsep=28pt,rowsep=24pt] & [name=NN]{\doublebox{$N,N$}} & [name=NH]{N,H} & [name=NT]{N,T}\\ & [name=HN]{H,N} & [name=HH]{\fbox{$H,H$}} &
[name=HT]{\fbox{$H,T$}}\\ & [name=TN]{T,N} & [name=TH]{\fbox{$T,H$}} & [name=TT]{\fbox{$T,T$}} %
\ncline[arrows=->]{NH}{NN} %
\ncline[arrows=->]{HN}{NN} %
\ncline[arrows=->]{NH}{HH} %
\ncline[arrows=->]{HN}{HH} %
\ncline[arrows=->]{NT}{HT} %
\ncline[arrows=->]{TN}{TH} %
\ncline[arrows=->]{TT}{TH} %
\ncline[arrows=->]{TH}{HH} %
\ncline[arrows=->]{HH}{HT} %
\ncline[arrows=->]{HT}{TT} %
\ncarc[arrows=->,arcangle=-30]{NT}{NN} %
\ncarc[arrows=->,arcangle=30]{TN}{NN}
  \end{psmatrix}\quad
  \begin{psmatrix}[nodesep=3pt,colsep=28pt,rowsep=24pt] 
& [name=NN]{\doublebox{$N,N$}} & [name=NH]{N,H} & [name=NT]{N,T}\\ 
& [name=HN]{H,N} & [name=HH]{\fbox{$\stackrel{\mbox{$\scriptscriptstyle H,H$}}{\mbox{$\scriptscriptstyle T,H$}}\ \ \stackrel{\mbox{$\scriptscriptstyle H,T$}}{\mbox{$\scriptscriptstyle T,T$}}$}} \\ %
& [name=TN]{T,N} & \ncline[arrows=->]{NH}{NN} %
\ncline[arrows=->]{HN}{NN} %
\ncline[arrows=->]{NH}{HH} %
\ncline[arrows=->]{HN}{HH} \ncarc[arrows=->,arcangle=30]{NT}{HH} %
\ncarc[arrows=->,arcangle=-30]{TN}{HH} \ncarc[arrows=->,arcangle=-30]{NT}{NN} %
\ncarc[arrows=->,arcangle=30]{TN}{NN}
  \end{psmatrix}
\]
  \caption{Matching pennies with hidden coins}
  \label{fig:match_pennies_2}
\end{figure}
\begin{example}[Prisoner's dilemma with communications]
  Adding communication to the prisoner's dilemma is easy.  Suppose there is a common knowledge among the prisoners that if they decide to fink or cooperate they will be put in the
  same cell and they will be able to communicate.  In the feasibility graph this means that one adds two arcs from $F,F$ to $Q,Q$, one for each prisoner.

\bigskip

  \hspace*{50pt}
  \begin{psmatrix}[colsep=2cm,rowsep=2cm] &[name=QQ]{Q,Q} & [name=QF]{Q,F}\\ &[name=FQ]{F,Q} & [name=FF]{F,F} \psset{nodesep=5pt}
    \ncline[arrows=<->,linewidth=.03,linecolor=blue,linestyle=dashed]{QQ}{QF} %
    \ncline[arrows=<->,linewidth=.03,linecolor=blue,linestyle=dashed]{FQ}{FF}
    \ncline[arrows=<->,linewidth=.08,linecolor=red,linestyle=dashed]{QQ}{FQ} %
    \ncline[arrows=<->,linewidth=.08,linecolor=red,linestyle=dashed]{QF}{FF}%
    \ncarc[arrows=<->,linewidth=.08,linecolor=red,linestyle=dashed, arcangle=-25]{FF}{QQ} %
    \ncarc[arrows=<->,linewidth=.03,linecolor=blue,linestyle=dashed, arcangle=25]{FF}{QQ}
   \end{psmatrix} 

\bigskip

\noindent and one gets the \emph{\fmdc} for both prisoners and the \emph{general \fmdc}.

\bigskip

  \hspace*{-30pt}
  \begin{psmatrix}[colsep=2cm,rowsep=2cm] &[name=QQ]{Q,Q} & [name=QF]{Q,F}\\ &[name=FQ]{F,Q} & [name=FF]{F,F} \psset{nodesep=3.5pt}
    \ncline[arrows=->,linewidth=.03,linecolor=blue]{QQ}{QF} %
     \ncline[arrows=->,linewidth=.03,linecolor=blue]{FQ}{FF} %
     \ncline[arrows=->,linewidth=.08,linecolor=red]{QQ}{FQ} %
     \ncline[arrows=->,linewidth=.08,linecolor=red]{QF}{FF} %
     \ncarc[arrows=->,linewidth=.08,linecolor=red, arcangle=-25]{FF}{QQ} %
    \ncarc[arrows=->,linewidth=.03,linecolor=blue, arcangle=25]{FF}{QQ} %
  \end{psmatrix}
  \begin{psmatrix}[colsep=2cm,rowsep=2cm] &[name=QQ]{Q,Q} & [name=QF]{Q,F}\\ &[name=FQ]{F,Q} & [name=FF]{F,F} \psset{nodesep=3.5pt}
    \ncline[arrows=->,linewidth=.03]{QQ}{QF} %
     \ncline[arrows=->,linewidth=.03]{FQ}{FF} %
     \ncline[arrows=->,linewidth=.03]{QQ}{FQ} %
     \ncline[arrows=->,linewidth=.03]{QF}{FF} %
     \ncline[arrows=->,linewidth=.03]{FF}{QQ} %
  \end{psmatrix}

\bigskip

\noindent In the last graph there is one SCC, namely $\{(Q,Q), (Q,F), (F,Q), (F,F)\}$.  This means that the prisoners will never be able to make their mind.
\end{example}

\begin{example}[A quest for the wonderland, the FD equilibria] In Example~\ref{ex:quest}, there are three equilibria (see Fig.~\ref{fig:wonderlang}).  In some cases the wonderland is a group
  of islands. Only island $D$ is a wonderland by itself, therefore the crews can decide to burn their ships or to sink them and stay there. Notice that both crews agree on
  what a wonderland is and that the decision is made collectively as one decision maker.
  \begin{figure}[ht]
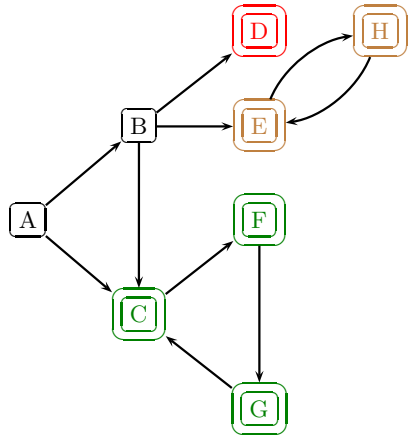

     \begin{psmatrix}[colsep=25pt,rowsep=15pt]
      &&&[name=D]{\color{red}\ovalbox{\ovalbox{D}}} & [name=H]{\color{brown}\ovalbox{\ovalbox{H}}}\\
      &&[name=B]{\ovalbox{B}} & [name=E]{\color{brown}\ovalbox{\ovalbox{E}}}\\
      &[name=A]{\ovalbox{A}} && [name=F]{\color{vertfonce}\ovalbox{\ovalbox{F}}} \\
      &&[name=C]{\color{vertfonce}\ovalbox{\ovalbox{C}}} \\
      &&&[name=G]{\color{vertfonce}\ovalbox{\ovalbox{G}}} \ncline[arrows=->,linewidth=.03]{A}{B} %
      \ncline[arrows=->,linewidth=.03]{A}{C} %
      \ncline[arrows=->,linewidth=.03]{B}{C} %
      \ncline[arrows=->,linewidth=.03]{C}{F} %
      \ncline[arrows=->,linewidth=.03]{F}{G} %
      \ncline[arrows=->,linewidth=.03]{G}{C} %
      \ncline[arrows=->,linewidth=.03]{B}{D} %
      \ncline[arrows=->,linewidth=.03]{B}{E} %
      \ncarc[arrows=->,linewidth=.03,arcangle=30]{E}{H} %
      \ncarc[arrows=->,linewidth=.03,arcangle=30]{H}{E} %
    \end{psmatrix}
     \centering
     \caption{The three equilibria for the quest for the wonderland}
    \label{fig:wonderlang}
  \end{figure}
\end{example}

\section{Choice}
\label{sec:choice}

Let us now look at choices.  Let us write \PX{} the set of non empty subsets of a grand set $X$.  A \emph{choice} is a
partial function $C: \PX "->" X$ with the following constraint\footnote{This means that the choice is consistent because
  it is an element of the set ; $`k$ stands for \emph{consistency}.}:

\medskip

\noindent \textbf{Condition $`k^+$:} \quad $C(A)\in A$.

\medskip

\noindent Write $Dom(C)$, the domain of $C$, i.e. the set of subsets on which $C$ is defined.  When we write $C(A)$, we
assume $A\in \Dom(C)$.  There is another constraint on $C$ called \emph{condition $`a^+$} which says that $C$ is stable
on subsets.

\medskip

\noindent \textbf{Condition $`a^+$:}\quad If $A\subseteq B$ and $C(B)\in A$ then $C(A) = C(B)$.

\medskip

\noindent A \emph{choice correspondence}, written $\Cc$, is a partial function, from \PX{} to \PX{} i.e. such that $\Cc:\PX "->" \PX$, with the constraint:

\medskip \noindent \textbf{Condition $`k$:} \quad $\Cc(A)\subseteq A$.

\medskip

\noindent This means that all the choices $\Cc(A)$ made by $\Cc$ are in $A$.  Since $\Cc$ can be partial, $Dom(\Cc)$ can be a strict subset of $X$. A set $Y$ of subsets of $X$ is a
\emph{$\cap$-semi-lattice} if the intersection of two subsets in $Y$ is in $Y$.
\begin{definition}[$\cap$-semi-lattice] A subset Y of \PX{} is a $\cap$-semi-lattice iff $A\in Y$ and $B\in Y$ implies $(A\cap B) \in Y$.
\end{definition}
Assume $Dom(\Cc)$ is a $\cap$-semi-lattice, we can state a new condition on $\Cc$.

\medskip

\noindent \textbf{Condition $`i$:}\quad If $x\in A$ and $x\in \Cc(B)$ then $x\in \Cc(A\cap B)$.  

In other words $A\cap \Cc(B) \subseteq \Cc(A\cap B)$.

\begin{proposition}
  If conditions $`k$ and $`i$ hold then \(\Cc(A) \cap \Cc(B) \subseteq \Cc(A\cap B).\)
\end{proposition}

For choice correspondence, condition $`a^+$ must be changed to take into account that the function does not produce an element but produces a subset.







\bigskip

\noindent \textbf{Condition $`a$:}\quad If $A\subseteq B$, if ${x\in A}$ and if ${x\in \Cc(B)}$, then ${x\in \Cc(A)}$.  

In other words, $A\subseteq B$ implies $A\cap \Cc(B) \subseteq \Cc(A)$.

\medskip

This can be rephrased as ``if $A\subseteq B$ then $A \cap \Cc(B) \subseteq \Cc(A)$''.  In particular, if $\Cc$ produces always a singleton i.e. $\Cc(A) = \{c_A\}$, which means that $\Cc$ is a
choice function, then \emph{condition $`a$} says that if $x\in A$ and if $x\in \{c_B\}$ then $x\in\{c_A\}$, in other words if $x\in A$ and if $x=c_B$ then $x=c_A$ which means
exactly if $c_B\in A$ then $c_A=c_B$, which is condition $`a^+$.   The following proposition can be stated.
\begin{proposition}
  Assume $\Dom(\Cc)$ is a $\cap$-semi-lattice, then
  \begin{enumerate}
  \item $\Cc$ satisfies \emph{condition~$`i$} implies $\Cc$ satisfies \emph{condition $`a$}.
  \item $\Cc$ satisfies \emph{condition~$`k$} and \emph{condition~$`a$}  implies $\Cc$ satisfies  \emph{condition $`i$}.
  \end{enumerate}
\end{proposition}
\begin{proof}
  For $1.$, if $A \subseteq B$ then $A\cap B=A$, then condition $`i$ implies condition~$`a$.

  For $2.$, one has $A\cap B\subseteq B$, condition $`k$ can be rephrased as $B\cap \Cc(B)=\Cc(B)$ and condition $`a$ is $A\cap B\cap \Cc(B) \subseteq \Cc(A\cap B)$ which is
  condition $`i$.
\end{proof}

\section{Acyclic Relations as Choice Correspondences}
\label{sec:EqAN_as_ch}

In this section we are going to show that any relation (not necessary an order) can be used as a choice correspondence. The only requirement is the existence of a sink.
Such a relation is naturally given by a \emph{\fmdc}.

Now we abstract the presentation of the previous sections, by considering only the pair $\langle \S, "->"\rangle$. One may consider $\S$ as a set of situations and $"->"$ as a
\emph{\fmdc}.  In other words, we forget about feasibility and desirability for a while and retain only~$"->"$.  From a choice point of view, $\S$ is the set of \emph{alternatives}.
A~\emph{choice maker} is then a combination of agents participating in an FD game.  $\Cc_{"->"}$ is the choice correspondence for the relation~$"->"$.  It associates with~$A$ the set
of sinks in $A$ for the relation $"->"$, more precisely:
\[\Cc_{"->"}(A) = \{ s\in A\mid `A s'\in A,\ s "->" s' \ "=>"\ s=s'\}.\] 
In Fig.~\ref{fig:corresp}, $\Cc_{"->"}(A)$ is made of the red triangles and the blue square and $\Cc_{"->"}(B)$ is made of the green circles and the blue square.

\begin{figure}[ht]
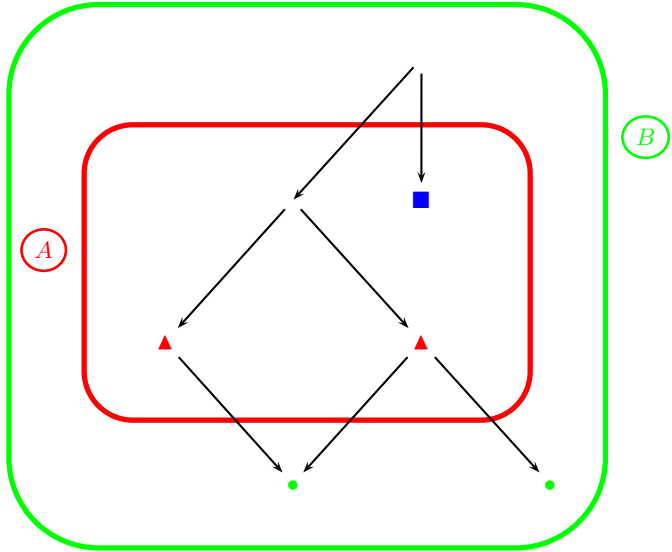

 \vspace*{20pt}
  \(
  \psframe[linewidth=2pt,framearc=.33,linecolor=red](8,2.7)(2,6.7) \psframe[linewidth=2pt,framearc=.33,linecolor=green](9,8.3)(1,1)
  \rput(9.5,6.5){\ovalnode[linewidth=1pt,linecolor=green]{B}{\color{green} B}} \rput(1.5,5){\ovalnode[linewidth=1pt,linecolor=red]{A}{\color{red} A}}
  \begin{psmatrix}[nodesep=3pt]
    &&&& [name=a] \\
    &&& [name=c] & [name=d]{\color{blue} \blacksquare}\\
    &&[name=e] {\color{red} \blacktriangle} && [name=f] {\color{red} \blacktriangle} \\
    &&&[name=g]{\color{green} \bullet} && [name=h]{\color{green} \bullet}
    \\
    \ncline[arrows=->]{a}{c} \ncline[arrows=->]{a}{d} \ncline[arrows=->]{c}{e} \ncline[arrows=->]{c}{f} \ncline[arrows=->]{e}{g} \ncline[arrows=->]{f}{g} \ncline[arrows=->]{f}{h}
  \end{psmatrix}
  \)\centering
  \caption{A relation and two choice correspondences}
  \label{fig:corresp}
\end{figure}
\begin{claim}
  $\Cc_{"->"}$ is a choice correspondence. Its domain is \[\Dom(\Cc_{"->"}) =\{A`: \PX \mid `E s`:A, `A s'`:A, s "->" s' \ "=>" s =s'\}.\]
\end{claim}
In other words, the domain of $\Cc_{"->"}$ is the set of sets which have at least a sink situation. 
\begin{claim}
  If $"->"$ is acyclic and $B$ is finite and not empty then $\Cc_{"->"}(B)$ is finite and non empty, in other words, $\Dom(\Cc_{"->"}) = \PX$.
\end{claim}
\begin{proof}
  A finite acyclic relation has always a sink.
\end{proof}
\begin{claim}
  $\Cc_{"->"}$ fulfills condition $`k$, i.e. $\Cc_{"->"}(A)\subseteq A$.
\end{claim}

\begin{claim}
  If $"->"$ is acyclic, then $\Cc_{"->"}$ fulfills condition $`a$.
\end{claim}
\begin{proof}
  If $A\subseteq B$, condition $`A s`:B, s"->"s' \ "=>"\ s=s'$ is stronger than ${`A s`:A, s"->"s' \ "=>" s=s'}$ hence $s`:A$ and $s`:\Cc_{"->"}(B)$ imply $s`:\Cc_{"->"}(A)$.
\end{proof}
This is illustrated by the diagram of Fig.~\ref{fig:corresp}.  The blue square belongs to $A$ and is a sink in $B$ then it is also a sink in $A$, in other words, ${\color{blue}
  \blacksquare}`:A$ and ${\color{blue} \blacksquare}`:\Cc_{"->"}(B)$ implies ${\color{blue} \blacksquare}`:\Cc_{"->"}(A)$.

\section{Evolutionary Dynamic Games}
\label{sec:evo}

In \cite{roux08:_conver_prefer_games}, we presented a game called \emph{Blink and you lose} due to Ren\'e Vestergaard.
\subsection{Blink and you lose}
\label{sec:blink} \emph{Blink and you lose} is a game played on a
simple graph with two undifferentiated tokens.  There are three
positions:
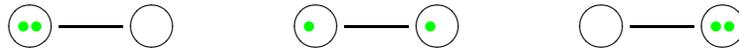
\[ \xymatrix { *++[o][F-]{\verte{`(!)}\verte{`(!)}}\ar@{-}[r] &
*++[o][F-]{\phantom{`(!)}\phantom{`(!)}}  } \qquad\qquad \xymatrix {
*++[o][F-]{\verte{`(!)}\phantom{`(!)}}\ar@{-}[r] &
*++[o][F-]{\verte{`(!)}\phantom{`(!)}}  } \qquad\qquad \xymatrix {
*++[o][F-]{\phantom{`(!)}\phantom{`(!)}}\ar@{-}[r] &
*++[o][F-]{\verte{`(!)}\verte{`(!)}}  }
\] There are two players, \emph{Left} and \emph{Right}.  The leftmost
position above is the winning position for \emph{Left} and the
rightmost position is the winning position for \emph{Right}.  In other
words, the one who owns both token is the winner.  Let us call the
positions $L$, $C,$ and $R$ respectively.  One plays by taking a token
on the opposite node.  The game has four tactics.

\subsubsection{A first tactic: Foresight}
\label{sec:ByL1}

A player realizes that she can win by taking the opponent's token
faster than the opponent can react, i.e., player \emph{Left} can
convert $C$ into $L$ by outpacing player \emph{Right}. Player
\emph{Right}, in turn, can convert $C$ into $R$. This version of the
game has two singleton equilibria: $L$ and $R$.  This is described by
the following feasibility.
\[
\begin{psmatrix}[colsep=1.2cm,rowsep=1.2cm,nodesep=5pt]
  & [name=L]{L} & [name=C]{C}& [name=R]{R}
  \ncline[arrows=->,linewidth=.03,linecolor=blue,linestyle=dashed]{C}{R}
  \psset{arrowscale=.5}\ncline[arrows=->,linewidth=.08,arrowscale=.6,linecolor=red,linestyle=dashed]{C}{L}
\end{psmatrix}
\]
desirability is
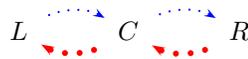
\[
\begin{psmatrix}[colsep=1.2cm,rowsep=1.2cm,nodesep=5pt]
  & [name=L]{L} & [name=C]{C}& [name=R]{R}
\ncarc[arrows=->,linewidth=.03,linecolor=blue,linestyle=dotted,arcangle=30]{C}{R}
\ncarc[arrows=->,linewidth=.03,linecolor=blue,linestyle=dotted,arcangle=30]{L}{C}
\psset{arrowscale=.5}
\ncarc[arrows=->,linewidth=.08,arrowscale=.6,linecolor=red,linestyle=dotted,arcangle=30]{C}{L}
\ncarc[arrows=->,linewidth=.08,arrowscale=.6,linecolor=red,linestyle=dotted,arcangle=30]{R}{C}
\end{psmatrix}
\]

\medskip

where $\begin{psmatrix}[colsep=1cm,arrowscale=.5]
  [name=C]& [name=R]
\ncline[arrows=->,linewidth=.08,arrowscale=.6,linecolor=red,linestyle=dotted,arcangle=30]{C}{R}
\end{psmatrix}$ is the desirability
for \emph{Left} and $\begin{psmatrix}[colsep=1cm]
  [name=C]& [name=R]
\ncline[arrows=->,linewidth=.03,linecolor=blue,linestyle=dotted,arcangle=30]{C}{R}
\end{psmatrix}$ is the desirability for
\emph{Right}.  The general \fmdc{} is then:
\[\begin{psmatrix}[colsep=1.2cm,rowsep=1.2cm,nodesep=5pt]
  & [name=L]{L} & [name=C]{C}& [name=R]{R}
 \ncline[arrows=->]{C}{L} 
 \ncline[arrows=->]{C}{R}
\end{psmatrix}
\] and one sees that
there are two equilibria: namely $L$ and $R$, which means that players
have taken both token and keep them.

\subsubsection{A second tactic: Hindsight}
\label{sec:ByL2}

A player, say \emph{Left}, analyzes what would happen if she does not
act. In case \emph{Right} acts, the game would end up in $R$ and
\emph{Left} loses.  As we all know, people hate to lose so they have
an aversion for a losing position. Actually \emph{Left} concludes
that she could have prevented the $R$ outcome by acting. In other
words, it is within \emph{Left}'s power to convert $R$
into~$C$. Similarly for player \emph{Right} from $L$ to~$C$.
\[
\begin{psmatrix}[colsep=1.2cm,rowsep=1.2cm,nodesep=5pt]
  & [name=L]{L} & [name=C]{C}& [name=R]{R}
  \ncline[arrows=->,linewidth=.03,linecolor=blue,linestyle=dashed]{R}{C}
  \psset{arrowscale=.5}\ncline[arrows=->,linewidth=.08,arrowscale=.6,linecolor=red,linestyle=dashed]{L}{C}
\end{psmatrix}
\]

We get the following general \fmdc{}:
\[\begin{psmatrix}[colsep=1.2cm,rowsep=1.2cm,nodesep=5pt]
  & [name=L]{L} & [name=C]{C}& [name=R]{R}
 \ncline[arrows=->]{L}{C} 
 \ncline[arrows=->]{R}{C}
\end{psmatrix}
\] 
where $C$ is a singleton equilibrium or an Abstract Nash Equilibrium.

\subsubsection{A third tactic: Omnisight}
\label{sec:ByL3}

The players have both hindsight and foresight, resulting in an FD game
\[
\begin{psmatrix}[colsep=1.2cm,rowsep=1.2cm,nodesep=5pt]
  & [name=L]{L} & [name=C]{C}& [name=R]{R}
\ncarc[arrows=->,linewidth=.03,linecolor=blue,linestyle=dashed,arcangle=30]{C}{R}
\ncarc[arrows=->,linewidth=.03,linecolor=blue,linestyle=dashed,arcangle=30]{L}{C}
\psset{arrowscale=.5}
\ncarc[arrows=->,linewidth=.08,arrowscale=.6,linecolor=red,linestyle=dashed,arcangle=30]{C}{L}
\ncarc[arrows=->,linewidth=.08,arrowscale=.6,linecolor=red,linestyle=dashed,arcangle=30]{R}{C}
\end{psmatrix}
\]
with one change-of-mind equilibrium covering all
outcomes thus, no singleton equilibrium (or Abstract Nash Equilibrium) exists.
\[
\begin{psmatrix}[colsep=1.2cm,rowsep=1.2cm,nodesep=5pt]
  & [name=L]{L} & [name=C]{C}& [name=R]{R}
\ncarc[arrows=->,arcangle=30]{C}{R}
\ncarc[arrows=->,arcangle=30]{L}{C}
\ncarc[arrows=->,arcangle=30]{C}{L}
\ncarc[arrows=->,arcangle=30]{R}{C}
\end{psmatrix}
\]

\subsubsection{A fourth tactic: Defeatism}
\label{sec:defeat}

One of the player, say \emph{Left}, acknowledges that she will be
outperformed by the other (\emph{Right} in this case).  She is so
terrified by her opponent that she returns the token when she has it.\footnote{This strategy was not presented in \cite{roux08:_conver_prefer_games}.}
This yields the following feasibility:
\[
\begin{psmatrix}[colsep=1.2cm,rowsep=1.2cm,nodesep=5pt]
  & [name=L]{L} & [name=C]{C}& [name=R]{R}
  \ncline[arrows=<-,linewidth=.03,linecolor=blue,linestyle=dashed]{R}{C}
  \psset{arrowscale=.5}\ncline[arrows=->,linewidth=.08,arrowscale=.6,linecolor=red,linestyle=dashed]{L}{C}
\end{psmatrix}
\]
We get the following general \fmdc{}
\[\begin{psmatrix}[colsep=1.2cm,rowsep=1.2cm,nodesep=5pt]
  & [name=L]{L} & [name=C]{C}& [name=R]{R}
 \ncline[arrows=->]{L}{C} 
 \ncline[arrows=->]{C}{R}
\end{psmatrix}
\] 
where $R$ is a singleton equilibrium or an Abstract Nash Equilibrium.

\subsection{Blink and you loose and evolutionary games}
\label{sec:evol}

\cite{nowak04:_evolut_dynam_of_biolog_games} present outcomes of evolutionary games of two strategies, like in Fig.~\ref{fig:Now-Sig}.  They call the first
outcome: \emph{dominance}, ($A$ vanishes, if $B$ is the best reply to both $A$ and $B$), for us this is the tactic \emph{defeatism} in \emph{Blink and you loose}.  They call the
second outcome \emph{bistability} (either $A$ or $B$ vanishes, depending on the initial mixture, if each strategy is the best response to itself), for us this is the tactic
\emph{foresight}.  They call the third outcome: \emph{coexistence} ($A$~and $B$ coexist in stable proportion, if each strategy is the best response to the other), for us this is
the tactic \emph{hindsight}.  They call the fourth outcome: \emph{neutrality} (the frequencies of $A$ and $B$ are only subject to random drift, if each strategy fares as well as the other
for any composition of the population and exhibit the same pictures), for this corresponds to the tactic \emph{omnisight}.   

\begin{figure}[t]
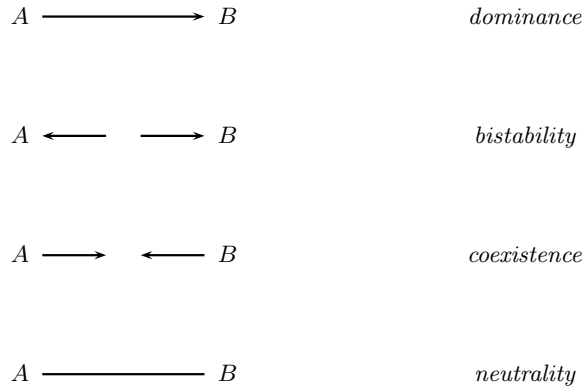

  \centering
  \[
\begin{psmatrix}[colsep=1.2cm,rowsep=1.2cm,nodesep=5pt]
  & [name=A1]{A} && [name=B1]{B}& {\qquad} & \emph{dominance}\\
  & [name=A2]{A} & [name=C2]{~} &  [name=B2]{B}&   {\qquad} & \emph{bistability}\\
  & [name=A3]{A} & [name=C3]{~} &  [name=B3]{B}&   {\qquad} & \emph{coexistence}\\
  & [name=A4]{A} &&  [name=B4]{B}&   {\qquad} & \emph{neutrality}\\
\ncline[arrows=->,linewidth=.03]{A1}{B1}
\ncline[arrows=->,linewidth=.03]{C2}{A2} \ncline[arrows=->,linewidth=.03]{C2}{B2}
\ncline[arrows=->,linewidth=.03]{A3}{C3} \ncline[arrows=->,linewidth=.03]{B3}{C3}
\ncline[linewidth=.03]{A4}{B4}
\end{psmatrix}
\]
  \caption{Evolutionary game dynamics (after Nowak and Sigmund)}
  \label{fig:Now-Sig}
\end{figure}

\subsection{Evolution of FD games}
\label{sec:ev_FDg}

FD games are a natural frame for evolution of games, especially of normal form games.  First recall that they abstract normal form games. Second they describe accurately the
process of evolution.  Consider a sequence of games describing an evolution.  Suppose that at step $n$ player $a$ is at a situation $s_n$.  At step $n+1$, he will move, if this is
possible, to a situation $s_{n+1}$ he can move to (feasible) and he wants to (desirable).  Therefore he will naturally follow an arc of the relation \emph{\fmdc}.  There will be
two kinds of interesting outcomes: either at some step, players reach a situation they cannot proceed further, this is an abstract Nash equilibrium, or players reach a strongly
connected component for the relation \emph{\fmdc}, in which the move forever without being able to leave it, this is an FD equilibrium.

\subsection{Matching pennies with hidden coins}
\label{sec:mphc}

One can consider the evolution of the game of Fig.~\ref{fig:match_pennies_2}.  One sees that it can have two evolutions, either to $N,N$ which is an abstract Nash equilibrium or
to the SCC which is an FD-equilibrium.

\medskip

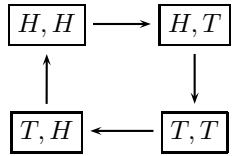
\begin{center}
  \begin{psmatrix}[nodesep=3pt,colsep=28pt,rowsep=24pt]
    & [name=HH]{\fbox{$H,H$}} & [name=HT]{\fbox{$H,T$}}\\
    & [name=TH]{\fbox{$T,H$}} & [name=TT]{\fbox{$T,T$}} %
    \ncline[arrows=->]{TT}{TH} %
    \ncline[arrows=->]{TH}{HH} %
    \ncline[arrows=->]{HH}{HT} %
    \ncline[arrows=->]{HT}{TT} %
  \end{psmatrix}
\end{center}

\paragraph{Acknowledgment}
I would like to thanks Franck Delaplace, St\'ephane Le~Roux and Ren\'e Vestergaard for discussions and suggestions.

\section{Conclusion}
\label{sec:concl}

This paper has shown the connection between game theory (more specially FD game theory which covers strategic game theory) choice models and evolutionary games, namely we have
shown that if the game has only abstract Nash equilibria, the function that yields the set of all the abstract Nash equilibria is a choice correspondence.


\appendix

\section*{Notions on directed graphs}

In this section we recall notions on relations which are also called directed oriented graphs and operations on them.
\begin{description}
\item[Graphs or relations] A \emph{directed graph} or a \emph{relation} $R$ on $A$ is a subset of $A\times A$. The elements of $A$ are called the \emph{nodes} of the graph and the pairs
  of $R$ are called the arcs of the graph.  One writes $x\,R\,y$ if $(x,y)`:R$.
\item[A relation is transitive] if $x\,R\,y$ and $y\,R\,z$ implies $x\,R\,z$.
\item[A relation is reflexive] if $x\,R\,x$ for all $x$.

\item[The transitive closure] of a graph or of a relation $R$ is the least relation $R^+$ which contains $R$.
\item[A directed acyclic graph] is a graph such that there exists no path from $x$ to $x$, in other words, there is no $x$, such that $x\,R^+\,x$.
\item[The transitive and reflexive closure] of a relation $R$ is the least transitive and reflexive relation $R^*$ which contains $R$.
\end{description}

\section*{Why is this article published only on the arXiv?}

This paper does present deep results from the mathematical point of view, like all the papers quoted in the
bibliography. Like previous ones \citep{Roux2006,roux08:_conver_prefer_games} it has been submitted to several journals and has
been rejected without having being actually read by any referee or editors, with only subjective arguments.  It should
be interesting to know what kinds of sociological arguments lie behind this rejection.  My view is that actually
scientists are reluctant to new ideas and new paradigms.  Thank to the \textsf{arXiv} system, this research can be made
available to a large community and to the future generation.

\end{document}